\documentclass[letterpaper, 10 pt, conference]{IEEEconf}  

\IEEEoverridecommandlockouts                              

\usepackage[style = ieee]{biblatex}
\def\BibTeX{{\rm B\kern-.05em{\sc i\kern-.025em b}\kern-.08em
    T\kern-.1667em\lower.7ex\hbox{E}\kern-.125emX}}

\usepackage{graphics} 
\usepackage{epsfig} 
\usepackage{mathptmx} 
\usepackage{times} 
\usepackage{amsmath} 
\usepackage{amssymb} 
\usepackage{amsthm} 
\usepackage{mathptmx}
\usepackage{caption}
\usepackage{subcaption}
\usepackage{comment}
\usepackage{mathtools}
\usepackage{xcolor}
\newcommand{\norm}[1]{\left\lVert#1\right\rVert}

\newtheorem{theorem}{Theorem}[section]
\newtheorem{corollary}{Corollary}[theorem]
\newtheorem{lemma}[theorem]{Lemma}
\newcommand{\R}{\mathbb{R}}

\DeclareMathAlphabet{\mathcal}{OMS}{cmsy}{m}{n}

\theoremstyle{remark}
\newtheorem*{remark}{Remark}

\theoremstyle{definition}

\pdfoutput=1
\usepackage{hyperref}
\hypersetup{
  pdfinfo={
    Title={title},
    Author={author},
  }
}
\usepackage{pdfpages}

\title{\LARGE \bf
On Stability Analysis of Power Grids with Synchronous Generators and Grid-Forming Converters under DC-side Current Limitation
}

\author{Sayan Samanta$^{1}$ and Nilanjan Ray Chaudhuri$^{1}$
\thanks{$^{1}$Sayan Samanta and Nilanjan Ray Chaudhuri are with The School of Electrical Engineering \& Computer Science, The Pennsylvania State University, University Park, PA 16802, USA.
        {\tt\small sps6260@psu.edu, nuc88@psu.edu}}%

}

\begin{document}

\maketitle
\thispagestyle{empty}
\pagestyle{empty}

\begin{abstract}
Stability of power grids with synchronous generators (SGs) and renewable generation interfaced with grid-forming converters (GFCs) under dc-side current limitation is studied. To that end, we first consider a simple 2-bus test system and reduced-order models to highlight the fundamental difference between two classes of GFC controls -- (A) droop, dispatchable virtual oscillator control (dVOC) and virtual synchronous machine (VSM), and (B) matching control. Next, we study Lyapunov stability and input-output stability of the dc voltage dynamics of class-A GFCs for the simple system and extend it to a generic system. Next, we provide a sufficiency condition for input-to-state stability of the 2-bus system with a class-B GFC and extend it for a generic system. Finally, time-domain simulations from a reduced-order averaged model of the simple test system and a detailed switched model of the GFC validate the proposed conditions.
\end{abstract}

\section{Introduction}
Although a lot of intellectual capital has been invested towards research on prospective grids with 100\% converter-based generation -- it is of the authors' opinion that such systems may not become a reality as far as bulk power systems are concerned. Bulk power grids of the near and even distant future are expected to have SGs
in them, since hydro, solar thermal, and nuclear power are all here to stay. 
Indeed, many studies have been performed on the penetration of converter-based resources in presence of SGs, e.g.  \cite{nrel_1_2015,eirgrid_2012} and references therein, which in spite of their obvious merit, lack analytical insights that are fundamental to identifying major challenges in modeling and control of such systems and develop new theories in solving them.

It is only in the recent past that these gaps and challenges were summarized in a comprehensive manner by Milano et-al \cite{drofler_2018_foundation_and_challenges}. Among the multitude of fertile areas of research that can be pursued to solve these challenges, we focus on the dynamics, stability, and control of the real power channel in such systems that primarily affects the dc-link voltages of converter-based renewable generation and frequency of the ac system.
To that end, we consider the GFC technology and it's interaction with SGs in a bulk power grid, where two classes of GFC controls -- (1) droop, dVOC and VSM \cite{divan_droop_1993,zhong_VSM,drofler_dvoc} -- we call it class-A, and (2) matching control \cite{drofler_matching_control_2018} -- we term it class-B, are compared.

Our research is motivated by two relatively new papers on this topic \cite{drofler_cdc_2017,drofler_journal_2020}. In \cite{drofler_cdc_2017}, modeling adequacy of such systems is established through singular perturbation theory -- our paper follows similar modeling guidelines. However, the control law assumed for governor action in SGs in \cite{drofler_cdc_2017} is not quite realistic. A more realistic turbine-governor dynamics is considered in a follow up paper \cite{drofler_journal_2020}. This paper showed some interesting findings on frequency of ac system and dc voltage  dynamics of GFCs in presence of dc-side and ac current limitations. It was demonstrated that in presence of dc-side current limit, the dc voltages of class-A GFCs can become unstable under large increase in load, while class-B GFCs demonstrate increased robustness in stability, since regulation of their ac side angle dynamics takes into account the dc voltage dynamics. However, analytical treatment of stability guarantees in presence of dc-side current limitations was reserved for future research. 
In addition, we feel that there
is a need to complement the efforts in contrasting the basics of class-A and class-B GFCs in these papers by presenting the characteristics of class-A GFCs in the converter power– dc voltage plane and also bringing more clarity on their
fundamental difference with the class-B counterpart. 

Thus motivated, the objectives of this paper are twofold -- (1) develop an understanding of the fundamental difference between the two classes of GFC controls; and (2) provide analytical guarantees of stability (for class-A and -B) and sufficiency conditions of instability (for class-A) in presence of dc-side current limitation, when such converters are connected to a power system with SG-based conventional generation. Presence of both ac and dc-side current limitations is considered out of scope for this work and will be reported in a future paper. Nevertheless, we have presented a discussion on this topic in Section~\ref{sec:StabDClimits}.C.

\section{ Classes of GFC Controls: Reduced-Order Model}\label{sec:GFC-class}
\par A typical circuit diagram of a GFC interfacing renewable resources is shown in Fig. \ref{fig_converter_model} whose dc bus is connected for example, to a PV solar unit or the dc side of ac-dc converter of a Type-4 wind turbine. Therefore, we restrict our focus to the dc to ac unidirectional power flow scenario, i.e. energy storage is excluded from our analysis. The notations associated with parameters and variables mentioned in this figure are standard and self-explanatory, see \cite{yazdani} for example. 

\vspace{-0pt}
\begin{figure}
        \centering
        \includegraphics[width=\linewidth]{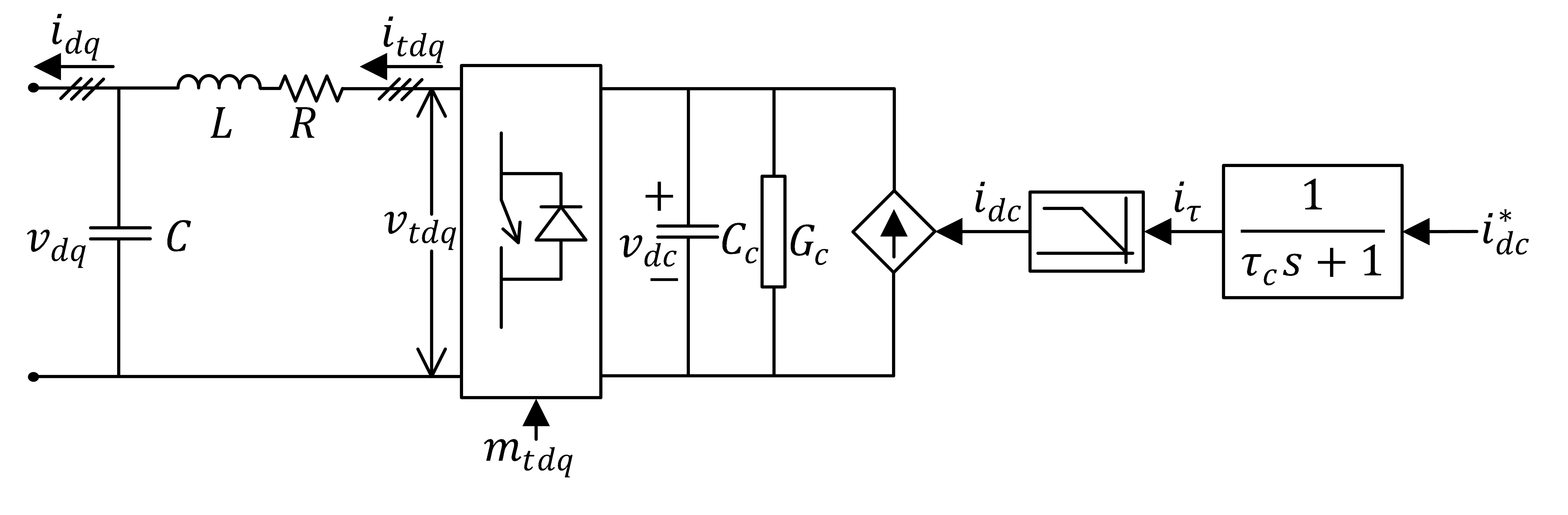}
        \vspace{-15pt}
        \caption{Circuit diagram of GFC.}%
        \vspace{-15 pt}
        \label{fig_converter_model}
\end{figure}
\vspace{-0 pt}

\par The GFC is controlled in a synchronously rotating $d-q$ reference frame whose angular frequency $\omega_c$ is imposed by the converter. The standard inner current control, albeit without any limits and voltage control loops, common across class-A and class-B, are shown in Fig.~\ref{fig_control_block}. It is the outer loops where the GFC control strategies differ -- for further details on class-A and class-B outer loops, the readers are referred to  \cite{drofler_cdc_2017,drofler_journal_2020}.

\vspace{-0pt}
\begin{figure}[ht]
        \centering
          \includegraphics[width=\linewidth]{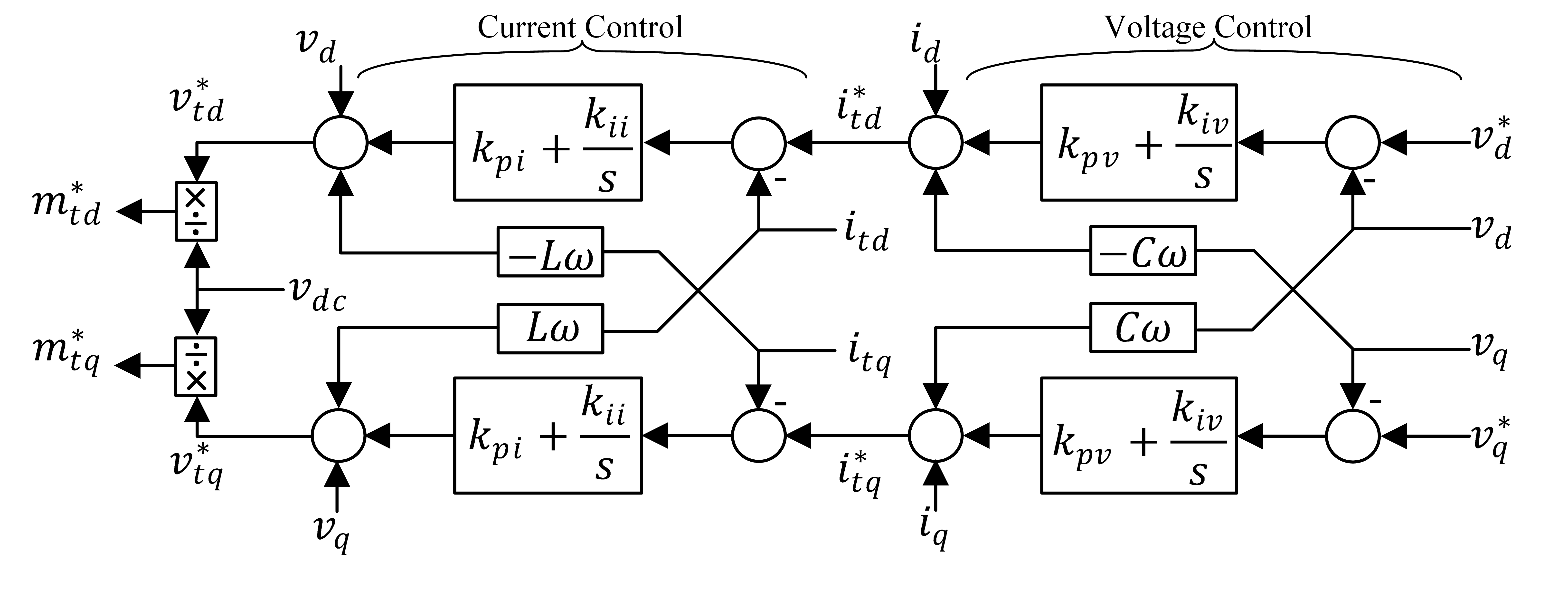}
            \caption{Block diagram of voltage and current control loops.}%
        \label{fig_control_block}
\end{figure}
\vspace{-0pt}

First, we consider a simple test system shown in Fig. \ref{fig_test_systeml} with one SG and one GFC. Based upon modeling adequacy study presented in \cite{drofler_cdc_2017}, we assume that the control loops shown in Fig. \ref{fig_control_block} track the references instantaneously, and the transmission dynamics and losses can be neglected. Moreover, since our focus is on the dynamics of the real power channel, we represent the dc-link dynamics of the GFC and the swing dynamics of the SG along with corresponding turbine-governor dynamics. Also, SGs are assumed to have adequate headroom to deliver any load change and a dc power flow assumption is made. Next, we present the reduced-order model of this system based on the time-scale separation of ac and dc dynamics  \cite{drofler_cdc_2017,drofler_journal_2020}.

\vspace{-0 pt}

\subsection{ Reduced-order Model of Class-A GFCs}\label{sec:classA}
\vspace{5pt}
\par Neglecting the time constant  $\tau_{c}$ of the dc energy source in Fig. \ref{fig_converter_model}, we can derive the test system model with class-A GFC shown in Fig. \ref{fig:block_diagram}(a):
\vspace{-0 pt}
	\begin{subequations}\label{eq:classAmodel}\small
	\begin{align}
 \dot v_{dc}  = \frac{1}{{C_c }}\left[ { - G_c v_{dc}  + sat\left( {k_c \left( {v_{dc}^*  - v_{dc}  } \right), i_{dc}^ {max}  } \right) - \frac{{P_c }}{{v_{dc} }}} \right] \\
 \dot \phi  = \underline \omega  _c-\underline \omega  _g; ~~ \underline \omega  _c = - d_{pc} \left( {P_c  - P_c^ *  } \right);~~~ \underline \omega  _g  = \omega _g  - \omega _g^ *   \\  
 \dot {\underline \omega _g}  = \frac{1}{{2H_g }}\left[ {P_{\tau g}  - P_g } \right]   \approx \frac{1}{{2H_g }}\left[ {P_{\tau g}  + b\phi - P_{Lg} } \right] \\ 
 \dot P_{\tau g}  = \frac{1}{{\tau _g }}\left[ {P_g^ *   - d_{pg} \underline \omega  _g  - P_{\tau g} } \right] 
 \end{align}    
\end{subequations}
\vspace{-0 pt}

where, $c,g,\tau g$: subscripts corresponding to GFC, SG, and turbine-governor, $*$: superscript for reference quantities, $v_{dc}$: dc-link voltage, $C_{c}$: dc-link capacitance, $G_c$: conductance representing dc-side losses, $k_c$: dc voltage droop constant, $i_{dc}^{max}$: dc-side current limit reflecting the capacity of the renewable resource, $P,~P_L$: real power output, load, $\phi$: angle difference between bus voltages of GFC and SG, i.e., $\phi = \theta_{c}-\theta_{g}$, $d_{pc}$: coefficient of droop/dVOC/VSM control, $\omega$: angular frequency, $H_g$: SG inertia constant, $\tau_g$: turbine time constant, $d_{pg}$: SG inverse governor droop, and $b$: transmission line susceptance.

\vspace{-0pt}
\begin{figure}[t]
        \centering
            \includegraphics[width=\linewidth]{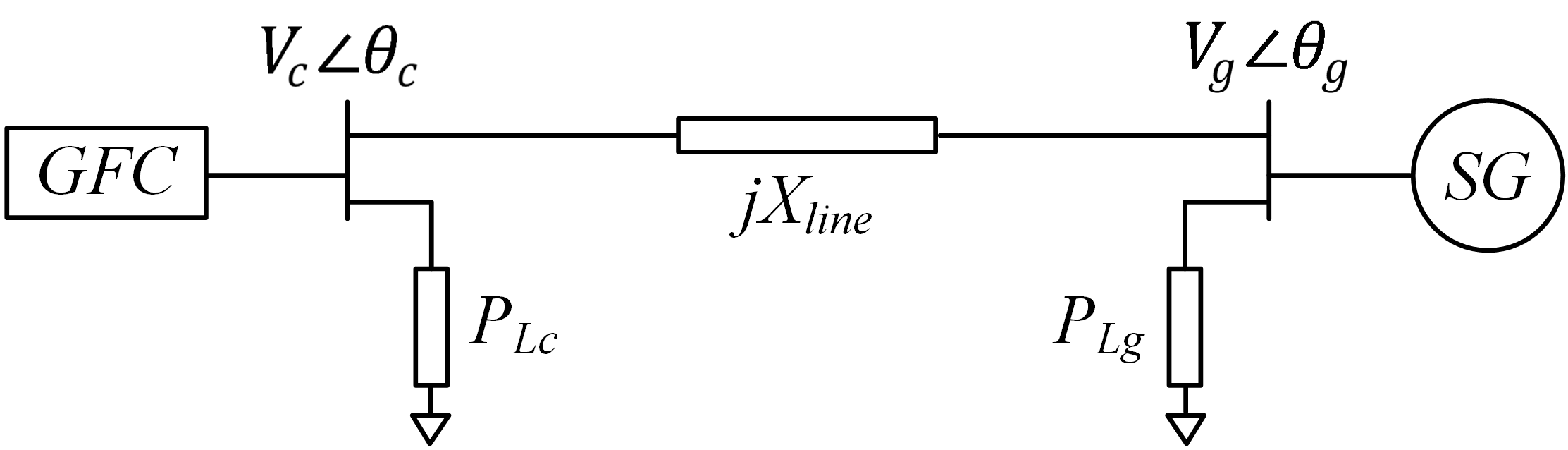}
            \caption{Single-line diagram of the test system.}%
        \label{fig_test_systeml}
        \vspace{-15 pt}

\end{figure}
\begin{figure}[ht]
         \includegraphics[width=\linewidth]{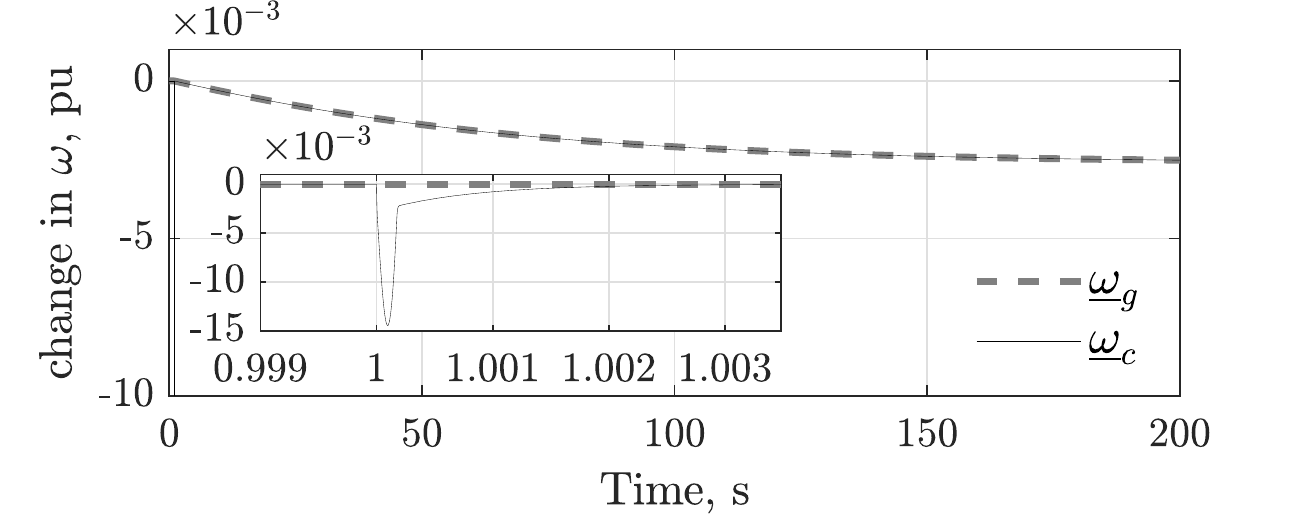}
     \caption{$\omega_{g}$ and $\omega_{c}$ for class-B GFC}
\label{fig_wg_wc_class_b}
\vspace{-15pt}
\end{figure}

\begin{figure*}[t]
        \centering
            \includegraphics[width=1.0\textwidth]{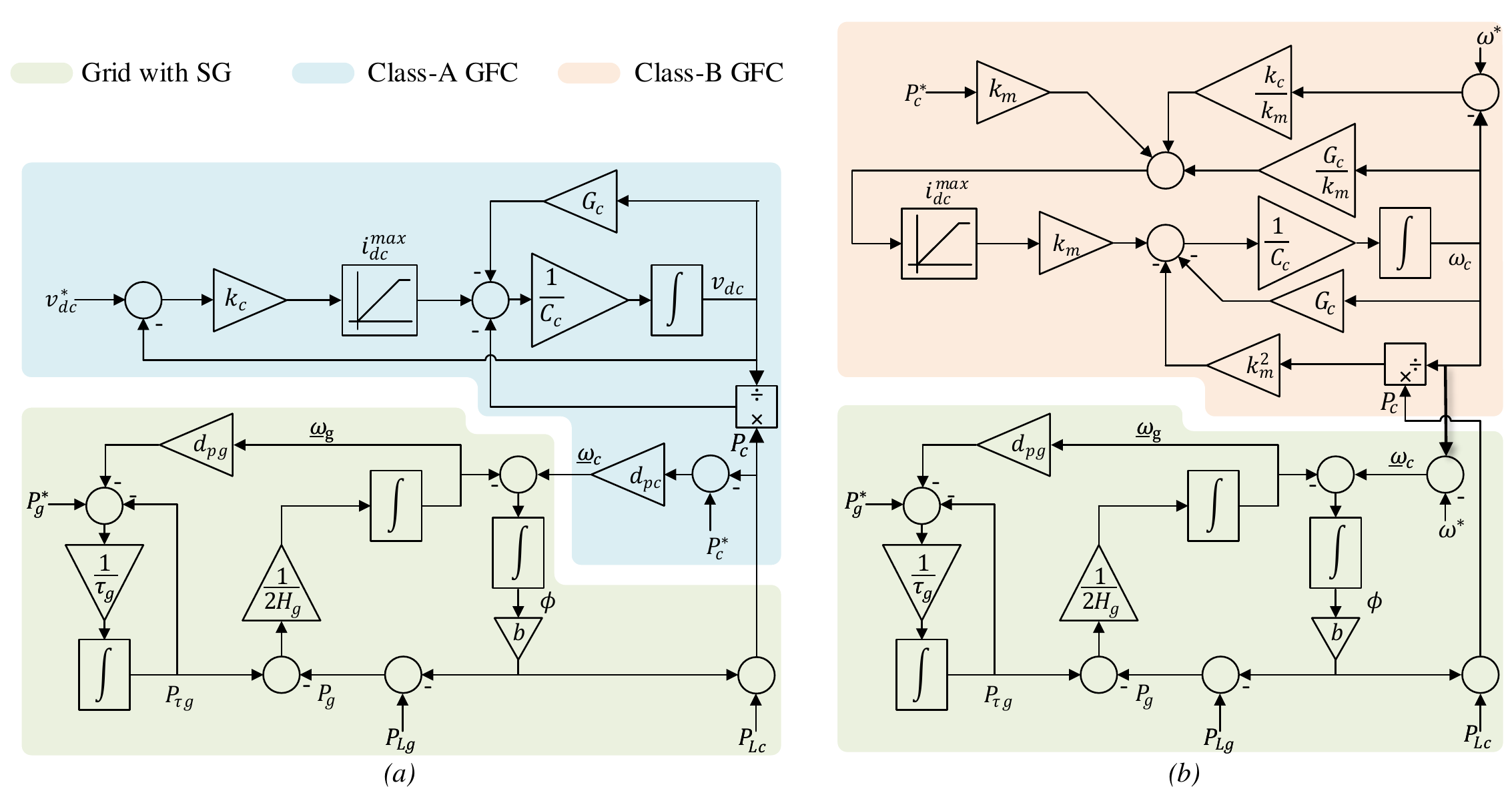}
            \caption{Block diagram of simplified mathematical models of the test system in Fig. \ref{fig_test_systeml} for: (a) Class-A GFCs (droop, dVOC, VSM control) and (b) Class-B GFCs (matching control): feedback path from GFC to grid is highlighted.}%
        \label{fig:block_diagram}
        \vspace{-0 pt}

\end{figure*}

\subsection{Reduced-order Model of Class-B GFCs} \label{sec:classB} Modifying (1a) by including the feedforward terms considered in matching control as in \cite{drofler_journal_2020}, we can write:
\begin{equation}\label{eq:classBvdc}
\small
C_c \dot v_{dc}  =  - G_c v_{dc}  + sat\left( {G_c v_{dc}  + \frac{{P_c^* }}{{v_{dc}^* }} + k_c \left( {v_{dc}^*  - v_{dc} } \right),  i_{dc}^{\max } } \right) - \frac{{P_c }}{{v_{dc} }}
\end{equation}
\normalsize
With matching control law $k_m v_{dc}  = \omega _c $ and $k_m v_{dc}^*  = \omega ^*  = 1$ pu, we can modify \eqref{eq:classBvdc} and \eqref{eq:classAmodel}(b)-(d) to derive the test system model with class-B GFC, which is shown in Fig. \ref{fig:block_diagram}(b).
In presence of matching control, as the angle dynamics is very fast compared to the machine dynamics, a reasonable assumption is $\omega _c  \approx \omega _g  \Rightarrow \underline \omega  _c  \approx \underline \omega  _g  = \omega _g  - \omega ^*$ \cite{drofler_journal_2020}. This can be shown through time-domain simulation of Fig. \ref{fig:block_diagram}(b) following a step change in $P_{Lc}$, which is highlighted in Fig. \ref{fig_wg_wc_class_b}. Since $\underline \omega_c$ and  $\underline \omega_g$ are indistinguishable with step change in $P_{Lg}$, it is not shown. With this approximation, we can write: 
\begin{equation}
\normalsize
\begin{array}{l}
 \frac{{C_c }}{{k_m^2 }}\underline {\dot \omega } _g  = sat\left( { - \frac{{k_{c} }}{{k_m^2 }}\underline \omega  _g , \underline P _c^{\max } } \right) + P_c^*  - P_c \frac{{\omega ^* }}{{\omega _g }} \\ 
 ~~~~~~\approx sat\left( { - d_{pc} \underline \omega  _g , \underline P _c^{\max } } \right) - \underline P _c  
 \end{array}
\end{equation}
With $ \frac{{C_c }}{{k_m^2 }} \approx 0$ as assumed in~\cite{drofler_journal_2020}, we can write:
\begin{equation}
\underline P _c  =  - sat\left( { d_{pc} \underline \omega  _g , \underline P _c^{\max } } \right)
\end{equation}
where, $d_{pc}  = \frac{{k_{c} }}{{k_m^2 }},\underline P _c^{\max }  = v_{dc}^* i_{dc}^{\max }  - G_c {v_{dc}^*}^2  - P_c^* ,\underline P _c  = P_c  - P_c^*. $
With total load in the system $P_L = P_{Lg} + P_{Lc}$ and power balance under nominal condition, i.e. $-P_g^* - P_c^* + P_L^* = 0$, we can write:
\begin{equation}\label{eq:ClassBmodel}
\begin{array}{l}
 \underline {\dot \omega } _g  = \frac{1}{{2H_g }}\left( {\underline P _{\tau g}  - sat\left( { d_{pc} \underline \omega  _g, \underline P _c^{\max } } \right) - \underline P _L } \right) \\ 
 \underline {\dot P} _{\tau g}  = \frac{1}{{\tau _g }}\left( { - \underline P _{\tau g}  - d_{pg} \underline \omega  _g } \right) 
 \end{array}
\end{equation}
where, $\underline P _{\tau g} = P _{\tau g} - P _{g}^*$ and $\underline{ P _L } = P _L - P _L^*$.

\subsection{Discussion on Fundamental Difference between Class-A and Class-B GFCs}
\label{sec: Difference} 
\vspace{5pt}
The model of class-A GFCs in \eqref{eq:classAmodel} is shown in a block diagram form in Fig. \ref{fig:block_diagram}(a). The most striking aspect of this class of control is that it merely acts as a buffer to adjust the frequency of its terminal voltage in order to deliver the power $P_c$ demanded by the system, which in turn affects the dc-link voltage dynamics. The converter has no direct control over $P_c$ and the dc-link dynamics does not have any `feedback mechanism' to alter it. Therefore, the stability of the dc-link voltage of class-A GFCs described by (\ref{eq:classAmodel}a) can be analyzed in isolation. Let, $v_{dc} = x>0 $, $v_{dc}^{*} = x^{*}$, $P_c = u >0$, and $(\bar{x}, \bar{u})$, $\bar x>0, \bar u >0$ be the equilibrium point. Also, assume $x^*$ is chosen such that the allowable maximum value of $x$ is $\tilde x^* = \frac{k_c}{(k_c + G_c)}x^*$, i.e. when $x \rightarrow \tilde x^*$, protective circuits will kick in and limit the dc voltage. Define, $y=x-\bar{x} \Rightarrow x = y+\bar{x}, v=u-\bar{u} \Rightarrow u = v+\bar{u}$.
Now, (\ref{eq:classAmodel}a) can be written as:
\begin{equation}\label{eqn_dc_dynamics_equilibrium}
    \dot{y} = \frac{1}{C_{c}}[-G_{c}(y+\bar{x})+sat(k_{c}(x^{*}-y-\bar{x}), i_{dc}^{max})-\frac{v+\bar{u}}{y+\bar{x}}]
\end{equation}
This equation is in the form $\dot{y} = f(y, v), ~y = h(y)$, where $f:D_y  \times D_v  \to \mathbb{R}$ is locally Lipschitz in $(y, v)$, $h:D_y  \to D_y$ is continuous in $(y, v)$, $f(0, 0) = 0$, and domains $D_y =(-\bar{x},~\tilde x^* - \bar{x})  \subset \mathbb{R}, D_v  \subset \mathbb{R}$ contain the origin. The equilibrium $(\bar{x}, \bar{u})$ satisfies the following equation:
    \[
   \bar{u} = \begin{cases}
f_1: -G_{c}\bar{x}^2 + k_{c} \bar{x}(x^* - \bar{x}), ~~if~~ |k_c(x^* - \bar{x})|\leq i_{dc}^{max} & \\
f_2: -G_{c}\bar{x}^2 + \bar{x}i_{dc}^{max},~~otherwise \end{cases}
    \]
Depending upon the value of $x$ where the maxima of $\bar{u}$ is found, we can get four types of characteristics in $x-u$ plane as shown in Fig. \ref{fig:pdc_vdc}. Out of these, the typical case is that in Fig.~\ref{fig:pdc_vdc}(a) -- going forward, unless otherwise mentioned, we will consider this characteristic. We note that for any given $\bar{u}$, there exists two equilibria $\bar{x}_1 \in \Omega_1 = \left[ {x_m ,\tilde x^* } \right)$ and $\bar{x}_2 \in \Omega_2 = 
\left( {0,\left. {x_m } \right]} \right.$, where $x_m = x^* - \frac{i_{dc}^{max}}{k_c}$.

\begin{figure}[ht]
        \vspace{-0pt}
        \centering
            \includegraphics[width=\linewidth]{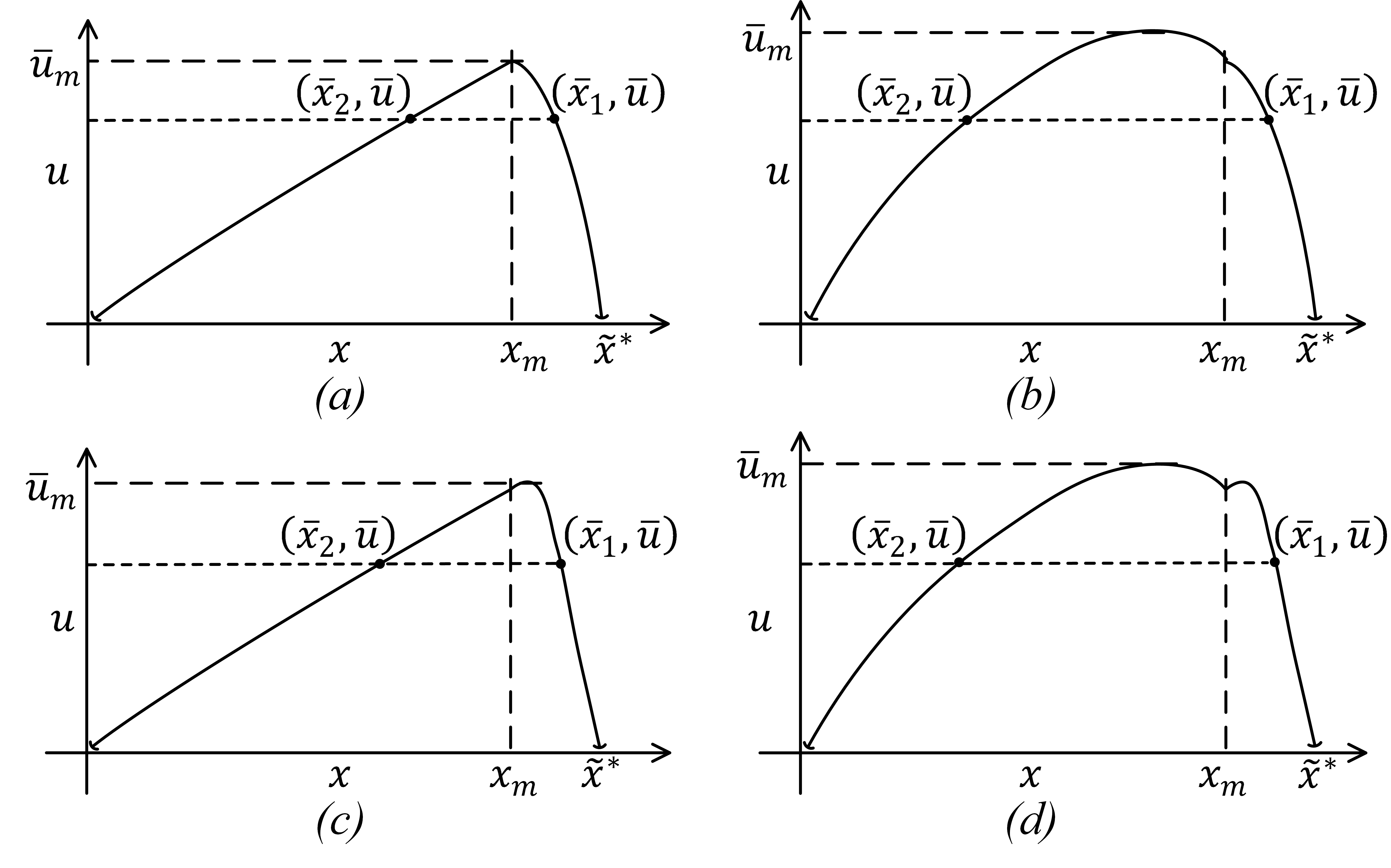}
            \caption{Four possible $u$ vs $x$ characteristics.}%
        \label{fig:pdc_vdc}
        \vspace{-0pt}
\end{figure}

In contrast, the model of class-B GFCs before reduction to the form in \eqref{eq:ClassBmodel} is shown in Fig. \ref{fig:block_diagram}(b). The most important difference with respect to its class-A counterpart is that it has a `feedback mechanism' from the dc-link voltage dynamics to the rest of the system that can alter the power $P_c$ demanded from the GFC. Thus, the stability of dc-link voltage can not be analyzed in isolation and a reduced-order model in \eqref{eq:ClassBmodel} is used for this purpose. Assuming $z = \left[ {\underline \omega  _g ~~\underline P _{\tau g} } \right]^T ,~~w =  - \underline P _L$, \eqref{eq:ClassBmodel} can be expressed as $\dot{z} = g(z,w), ~g: \R^{2}\times\R \rightarrow \R^{2}$, where $g$ is locally Lipschitz in $(z, w)$, and $g(0, 0) = 0$.
From a control design point of view, the fundamental difference between class-A and B can also be perceived as their dependency on either the ac or dc quantities \cite{gao2020gridforming,tayyebi2020hybrid}.
\vspace{5 pt}

\section{Stability Analysis in Presence of dc-side Current Limits}\label{sec:StabDClimits}
We focus on stability analysis of $v_{dc}$ in presence of dc-side current limits. To that end, we establish the following: (1)  Lyapunov stability and region of attraction (ROA) for class-A and class-B GFCs, (2) Sufficiency condition for input-output stability for class-A GFCs, (3) Sufficiency condition for instability of class-A GFCs, and (4) Sufficiency condition for input-to-state stability for class-B GFCs. 
\subsection{Stability Analysis of Class-A GFC}\label{sec:StabClassA}
We first focus on the reduced model of the 2-bus test system in Fig.~\ref{fig_test_systeml} and present the following Theorems and Lemmas.
\begin{theorem}\label{Th:ClassALyap}
For class-A GFCs, the equilibrium $\bar{x}_1$ is asymptotically stable with ROA $\mathcal{R_A} = \left( {\bar x_2 ,\left. {\tilde x^* } \right)} \right.$.
 \end{theorem}
 
 \begin{proof}
 Corresponding to the domain $\Omega_1$ for $x$,  \eqref{eqn_dc_dynamics_equilibrium} can be rewritten as:
 \begin{equation}\label{eq:ClassAvDC}
\dot y = \frac{1}{{C_c }}\left[ { - G_c y - k_c y + \frac{{\bar u}}{{\bar x_1 }}\frac{y}{{y + \bar x_1 }} - \frac{v}{{y + \bar x_1 }}} \right]
 \end{equation}
 where, $y \in \Tilde{D}_y = [x_m -\bar{x}_1,~\tilde x^* - \bar{x}_1)  \subset D_y, ~v \in D_v$.
Choosing a Lyapunov function $V_1 = \frac{{C_c }}{2}y^2 ,y \in \tilde D_y$, we can write for the unforced system: $\dot V_1 = \left[ { - \left( {G_c  + k_c } \right) + \frac{{\bar u}}{{\bar x_1 }}\frac{1}{{y + \bar x_1 }}} \right]y^2 $. It can be shown that $\dot V_1$ is negative definite, if $y > \tilde x^* - 2\bar{x}_1 $. In the most typical case as in Fig.~\ref{fig:pdc_vdc}(a), $\bar{x}_1 > x_m > \frac{\tilde x^*}{2}$, which satisfies this condition. Therefore, $\bar{x}_1$ is asymptotically stable $\forall x \in \Omega_1$.   

To establish the ROA of $\bar x_1$, we  analyze Lyapunov stability of $\bar x_2\in \Omega_2$ shown in Fig.~\ref{fig:pdc_vdc}(a) with the same $\bar{u}$. To that end, we can rewrite \eqref{eqn_dc_dynamics_equilibrium} with $v = 0$ as: $\dot y = \frac{1}{{C_c }}\left[ { - G_c \left( {y + \bar x_2 } \right) + i_{dc}^{\max }  - \frac{{\bar u}}{{y + \bar x_2 }}} \right],~~y \in \bar D_y  = (  - \bar x_2 ,x_m  - \bar x_2 ]$. Choosing a continuously differentiable function $
V_2 = \frac{{C_c }}{2}\left[ {\bar x_2^2  - \left( {y + \bar x_2 } \right)^2 } \right], ~~y \in \bar D_y$, s.t. $V_2\left( 0 \right) = 0$. We choose a ball $B_r = \left\{ {y~ \in \R~~ |~~ |y| \leq r } \right\}$ and define set $U = \left\{ {y~ \in B_r~~ |~~ V_2 > 0 } \right\}$ -- note that $
U \subseteq \left( { - \bar x_2 ,\left. 0 \right)} \right.$. 
Therefore, we can choose $y(0) = y_0 \in U$ arbitrarily close to the origin s.t. $V_2(y_0)>0$. Also, $\dot V_2 > 0, \forall y~\in U$, if $y < \frac{{i_{dc}^{\max } }}{G_c} - 2\bar x_2 $. Taking into account the typical characteristics in Fig.~\ref{fig:pdc_vdc}(a) and analyzing local maxima of $f_2$, we can write $\frac{{i_{dc}^{\max } }}{{2G_c}} > \bar x_2$. Therefore, $\dot V_2 > 0, \forall y~\in U$, which provides a sufficiency condition for instability of $\bar{x}_2$ following Chetaev's theorem \cite{khalil}. This implies that $x(t)$ with any initial value $x(0) = x_0\in \left( {0,\left. {\bar x_2 } \right)} \right. \subset \Omega _2 $ will move away from $\bar x_2$ and reach $0$.

Next, choosing a continuously differentiable function $V_3 = \frac{{C_c }}{2}\left[ {\left( {y + \bar x_2 } \right)^2 - \bar x_2^2} \right], ~~y \in \bar D_y$, s.t. $V_3\left( 0 \right) = 0$ -- it is easy to follow similar arguments and show that $V_3 > 0, \forall y \in \left( { 0 ,\left. x_m - \bar x_2 \right]} \right. \subset \bar D_y$. This implies that $x(t)$ with any initial value $x_0\in \bar \Omega _2 = \left( {\bar x_2,\left. {x_m } \right]} \right. \subset \Omega _2 $ will move away from $\bar x_2$ and reach $x_m$.

We define $\mathcal{R_A} = \bar \Omega _2 \bigcup \Omega_1 = \left( {\bar x_2,\left. {\tilde x^*} \right)} \right.$, which is the largest open, connected, invariant set in $\Omega _2 \bigcup \Omega_1$, such that $\mathop {\lim }\limits_{t \to \infty } x(t) = \bar x_1 ,\forall x\left( 0 \right) \in \mathcal{R_A} $. This implies $\mathcal{R_A}$ is the ROA for equilibrium $\bar x_1$ of class-A GFCs. 
\end{proof}


\begin{corollary}\label{Corr:ClassAExpStab}
For class-A GFCs, the equilibrium $\bar{x}_1$ is exponentially stable in $\Omega_1$. 
\end{corollary}

\begin{proof}
As mentioned in Theorem \ref{Th:ClassALyap}, the chosen Lyapunov function is $V_1 = \frac{{C_c }}{2}y^2 = \frac{{C_c }}{2}|y|^2 ,y \in \tilde D_y$. Also, $\dot V_1 \le \left[ { - \left( {G_c  + k_c } \right) + \frac{{\bar u}}{{\bar x_1 }}\frac{1}{{x_m }}} \right]|y|^2, y \in \tilde D_y$. Since, $m =  - \left( {G_c  + k_c } \right) + \frac{{\bar u}}{{\bar x_1 }}\frac{1}{{x_m }} < 0$, it satisfies all conditions in Theorem $4.10$ in \cite{khalil}, and therefore $\bar{x}_1$ is exponentially stable in $\Omega_1$. 
\end{proof}

\begin{theorem}\label{Th:ClassAInOut}
The dc voltage dynamics of class-A GFCs described in \eqref{eqn_dc_dynamics_equilibrium} is small-signal finite-gain $\mathcal{L}_p$ stable $\forall p \in [1, \infty]$, if $y(0) = y_0 \in \left\{ {|y|\leq r} \right\} \subset \tilde D_y, ~r>0$. Also, for a $r_v>0$, s.t. $\left\{ {|v|\leq r_v} \right\} \subset D_v, ~r_v>0, \forall v \in \mathcal{L}_{pe}$ with $\mathop {\sup |v|}\limits_{0 \le t \le \tau }  \le \min \left\{ {r_v ,|m|x_m r} \right\}$, the output $y(t)$ is bounded by the following relation $\norm{y_\tau}_{\mathcal{L}_p} \leq \frac{\norm{v_\tau}_{\mathcal{L}_p}}{|m|x_m}+\beta ~~\forall \tau \in [0, \infty)$, where $\beta = |y_0|,~~ if~~ p = \infty,$ and $\left( {\frac{C_c}{{p|m|}}} \right)^{\frac{1}{p}} |y_0 |,~~if~~p \in \left[ {1,\left. \infty  \right)}. \right.$ 
\end{theorem}

\begin{proof}
We proved that $y = 0$ is exponentially stable in $\tilde D_y$ in Corollary~\ref{Corr:ClassAExpStab}. With Lyapunov function $V_1 = \frac{{C_c }}{2}y^2 = \frac{{C_c }}{2}|y|^2$, we have $\dot V_1 \le -|m||y|^2 , ~\left| {\frac{{\partial V_1 }}{{\partial y}}} \right| = C_c \left| y \right|, ~\forall y \in \tilde D_y$. Also, $\left| {f\left( {y,v} \right) - f\left( {y,0} \right)} \right| \le \frac{1}{{C_c x_m }}\left| v \right|,\left| {h\left( {y,v} \right)} \right| = \left| y \right|,~\forall y \in \tilde D_y,~\forall v\in D_v $. This satisfies all conditions in Theorem $5.1$ in \cite{khalil} and proves the conditions for input-output stability and bound on output.
\end{proof}

\begin{theorem}\label{Th:ClassAInstab}
For class-A GFCs, the equilibrium $\bar{x}_1 \in \Omega_1$ of \eqref{eqn_dc_dynamics_equilibrium} with $v=0$ is unstable if $\bar u >  -G_{c}(y + \bar x_1)^{2} + (y + \bar x_1) sat(k_{c}(x^* - y - \bar x_1), i_{dc}^{max})$ for any $y \in [-r,~0)$, where $r = \min \left\{ {\bar x_1 ,\tilde x^*  - \bar x_1 } \right\}$.
 \end{theorem}
 
 \begin{proof}
The unforced system can be expressed as $\dot{y} = \frac{1}{C_{c}}[-G_{c}(y+\bar{x_1})+sat(k_{c}(x^{*}-y-\bar{x_1}), i_{dc}^{max})-\frac{\bar{u}}{y+\bar{x_1}}],~\forall y \in D_y \subset \R $.
Define a continuously differentiable function, $V_4:D_y\rightarrow \R $, $V_4(y) = \frac{1}{2}C_{c}[\bar{x}_1^{2}-(y+\bar{x}_1)^{2}]$ such that $V_4(0) = 0$.
Choose $r \in (0, \min \left\{ {\bar{x}_1,\tilde x^{*}-\bar{x}_1} \right\}]$ such that the ball $B_{r} = \{y \in \R | ~|y| \leq r\}$, $B_{r} \subset D_y$. Define, $U = \{y \in B_{r} | V_4(y) > 0 \}$, implying $U=[-r,~0)$. Choose $y_0$ in the interior of $U$ $\implies$ $y_{0} < 0$. Hence, $V_4(y_{0}) > 0$ for any such $y_{0}$ arbitrarily close to the origin.  
Now, derivative of $V_4$ along the trajectory of $y$ is:
\par $\dot{V_4} = G_{c}(y+\bar{x}_1)^{2} - (y+\bar{x}_1)sat(k_{c}(x^{*}-y-\bar{x}_1), i_{dc}^{max}) + \bar u$
\par According to Chetaev's theorem \cite{khalil}, the sufficiency condition  for instability is $\dot V_4 >0, ~\forall y \in U$, which proves the theorem.
\end{proof}
\vspace{-2pt}
Now, we extend these proofs for a generic system with $m_1$ SGs and $n_1$ class-A GFCs, and introduce the following Corollary.

\begin{corollary}
Theorems \ref{Th:ClassALyap},\ref{Th:ClassAInOut},\ref{Th:ClassAInstab} and Corollary~\ref{Corr:ClassAExpStab} hold for any generic system.
\end{corollary}
\begin{proof}
Discussions from Section~\ref{sec: Difference} establish that the stability properties investigated in these theorems are independent of the systems as long as the GFC-level assumptions taken in Section \ref{sec:GFC-class} hold. Therefore, these theorems and the corollary hold individually for each of the $n_1$ class-A GFCs.
\end{proof}

\subsection{Stability Analysis of Class-B GFC}\label{sec:StabClassB}
In this section, we first analyze the stability of class-B GFCs for the 2-bus system shown in Fig.~\ref{fig_test_systeml} and present the following lemma and theorem.

\begin{lemma}\label{Lemma:ClassBLyap}
For class-B GFCs, the equilibrium $z = 0$ is globally asymptotically stable $\forall \; d_{pg}, d_{pc}$ $>0$.
\end{lemma}
\begin{proof}
For unforced system, $w = -\underline{P}_{L} = 0$. Choose Lyapunov function with $d_{pg} > 0$,
$V_5 = H_{g}\underline{\omega}_{g}^{2} + \frac{\tau_{g}}{2d_{pg}}\underline{P}_{\tau g}^{2}$.
\[    \Rightarrow \dot{V}_5 = - \frac{\underline{P}_{\tau g}^{2}}{d_{pg}}-\underline{\omega}_{g} sat(d_{pc}\underline{\omega}_{g},\underline P _c^{\max }) 
\]
Here, 
$\underline{\omega}_{g} sat(d_{pc}\underline{\omega}_{g}, \underline P _c^{\max }) > 0 ~~ \forall \; \underline{\omega}_{g} \in \R - \left\{ {0} \right\}, d_{pc} > 0 $. Thus, $\dot{V}_5$ is negative definite and radially unbounded $\forall \; d_{pg}, d_{pc} > 0 $. Therefore, the origin is globally asymptotically stable when this condition is satisfied.
\end{proof}
\begin{remark}
We observe that the ROA for $y = 0$ corresponding to the equilibrium $\bar x_1$ of class-A GFCs is limited to $y \in (\bar x_2 - \bar x_1, \tilde x^* - \bar x_1)$, while the same for $z = 0$ of class-B GFCs is $\R^2$. Also, equilibrium $\bar x_2$ of class-A GFCs is unstable.
\end{remark}

\begin{theorem}\label{Th:InputStateClassB}
The reduced-order model~\eqref{eq:ClassBmodel} is input-to-state stable with class $\mathcal{KL}$ function $\beta$ and class $\mathcal{K}$ function $\gamma \left( {\left| w \right|} \right) = c\max \left\{ {\chi _1 \left( {\left| w \right|} \right),\chi _2 \left( {\left| w \right|} \right)} \right\},~c>0$ for piecewise continuous $w(t)$  that is bounded in $t,~\forall t\geq0$ implying $\left\| {z(t)} \right\| \le \beta \left( {\left\| {z(t_0 )} \right\|,t - t_0 } \right) + \gamma \left( {\mathop {\sup }\limits_{\tau  \ge t_0 } | {w\left( \tau  \right)} |} \right),\forall t \ge t_0$ -- where, $\chi _1 \left( {\left| w \right|} \right) = \frac{{\underline P _c^{\max } }}{{d_{pc} }}\tanh ^{ - 1} \left( {\frac{{|w|}}{{\theta \underline P _c^{\max } }}} \right) $ and $\chi _2 \left( {\left| w \right|} \right) = \left[ {\frac{|w|d_{pg}}{{\theta }} \chi _1 \left( {\left| w \right|} \right)} \right]^{\frac{1}{2}}$, $~\forall w\in (-\theta \underline P _c^{\max } ,~\theta \underline P _c^{\max } ),~\underline P _c^{\max }\in \mathbb{R}_{>0}, ~0<\theta<1$, $\mathbb{R}_{>0}$ : \text{positive real space}.

\end{theorem}
\begin{proof}
In Lemma~\ref{Lemma:ClassBLyap}, it is shown that $\dot{z} = g(z,0)$ is globally asymptotically stable. It can be shown that the Lyapunov function $V_5(z)$ satisfies the following inequalities: $\lambda _{\min } \left( Q \right)\left\| z \right\|_2^2  \le V_5 \left( z \right) \le \lambda _{\max } \left( Q \right)\left\| z \right\|_2^2$, which implies $\alpha _1 \left( {\left\| z \right\|} \right) \le V_5 \left( z \right) \le \alpha _2 \left( {\left\| z \right\|} \right)$, where $\alpha_1$ and $\alpha_2$ are class $\mathcal{K_\infty}$ functions and
$Q = 
\begin{bmatrix}
H_{g} & 0\\
0 & \frac{\tau_{g}}{2d_{pg}}
\end{bmatrix}$.
For $0<\theta<1$, we can write:
\begin{multline*}
\dot{V_5} =  - \frac{\underline{P}_{\tau g}^{2}}{d_{pg}}-\underline{\omega}_{g} sat(d_{pc}\underline{\omega}_{g}, \underline P _c^{\max }) + w\underline{\omega}_{g}
\\ \leq -(1- \theta)\left(\frac{\underline{P}_{\tau g}^{2}}{d_{pg}}+\underline{\omega}_{g} sat(d_{pc}\underline{\omega}_{g}, \underline P _c^{\max })\right)
\\ -\theta\left(\frac{\underline{P}_{\tau g}^{2}}{d_{pg}}+\underline{\omega}_{g} sat(d_{pc}\underline{\omega}_{g}, \underline P _c^{\max })\right)+ \left|w\right|\left|\underline{\omega}_{g}\right|
\end{multline*}
Let us define, $W = (1- \theta)\left(\frac{\underline{P}_{\tau g}^{2}}{d_{pg}}+\underline{\omega}_{g} sat(d_{pc}\underline{\omega}_{g}, \underline P _c^{\max })\right)$, which is a positive definite function in $\R^{2}$. Now, define $\Gamma = -\theta(\frac{\underline{P}_{\tau g}^{2}}{d_{pg}}+\underline{\omega}_{g} sat(d_{pc}\underline{\omega}_{g}, \underline P _c^{\max }))+ \left|w\right|\left|\underline{\omega}_{g}\right|$. The term $\Gamma$ will be $\leq$ 0 if $\left|\underline{\omega}_{g}\right| \geq \frac{{\underline P _c^{\max } }}{{d_{pc} }}\tanh ^{ - 1} \left( {\frac{{|w|}}{{\theta \underline P _c^{\max } }}} \right) = \chi _1 \left( {\left| w \right|} \right)$ or $\left|\underline{\omega}_{g}\right| \leq \chi _1 \left( {\left| w \right|} \right)$ and $\left|\underline P_{\tau g}\right| \geq \left[ {\frac{|w|d_{pg}}{{\theta }} \chi _1 \left( {\left| w \right|} \right)} \right]^{\frac{1}{2}} = \chi _2 \left( {\left| w \right|} \right)$. This condition implies $\left\| z \right\|_\infty   \geq \max \left\{ {\chi _1 \left( {\left| w \right|} \right),\chi _2 \left( {\left| w \right|} \right)} \right\} = \rho \left( {\left| w \right|} \right)$. So,
\[
\dot{V_5} \leq -W,~~ \forall \norm{z}_{\infty} \geq \rho (|w|)
\]
Here, $\rho(|w|)$ is a class $\mathcal{K}$ function with $w\in (-\theta \underline P _c^{\max } ,~\theta \underline P _c^{\max } )$. Since, $\underline P _c^{\max }\in \mathbb{R}_{>0}$, we contend that the above holds $\forall (z, w) \in \mathbb{R}^2\times\mathbb{R}$. Therefore, we have satisfied all conditions of input-to-state stability per Theorem $4.19$ in \cite{khalil}. 

Now, we need to define class  $\mathcal{K}$ function $\gamma = \alpha_{1}^{-1} \circ \alpha_{2} \circ \rho$. It can be shown that $\gamma \left( {\left| w \right|} \right)  = \sqrt {\frac{{\lambda _{\max } \left( Q \right)}}{{\lambda _{\min } \left( Q \right)}}} \rho \left( {\left| w \right|} \right) = c\rho \left( {\left| w \right|} \right)$.

\end{proof}


Next, we extend these proofs for a generic system with $m_1$ SGs, $n_1$ class-B GFCs, and $p_1$ load buses. We assume that the center-of-inertia (COI) of this system is representative of its average frequency dynamics and the corresponding frequency $\omega_{COI} \approx \omega_{ci}, ~\forall i=1,2,\dots,n_1$. Following the same approach as in Section~\ref{sec:classB}, we can present the reduced-order model of this system:
\begin{equation}\label{eq:ClassB_COI}
\begin{array}{l}
 \underline {\dot \omega } _{COI}  = \frac{1}{2{H_T }}\left[ {\underline P _{\tau gT}  - \sum\limits_{i = 1}^{n_1 } {sat\left( {d_{pci} \underline \omega  _{COI} ,  \underline P _{ci}^{\max } } \right) - \underline P _{LT} } } \right] \\ 
 \underline {\dot P} _{\tau gT}  = \frac{1}{{\tau _{gT} }}\left[ { - \underline P _{\tau gT}  - d_{pgT} \underline \omega  _{COI} } \right] \\ 
 \end{array}
\end{equation}
Here, $H_T  = \sum\limits_{i = 1}^{m_1 } {H_{gi} },~ P_{\tau gT}  = \sum\limits_{i = 1}^{m_1 } {P_{\tau gi} },~
P_{gT}^*  = \sum\limits_{i = 1}^{m_1 } {P_{gi}^* },~ d_{pgT}  = \sum\limits_{i = 1}^{m_1 } {d_{pgi} },~ \tau _{gi}  = \tau _{gT} \forall i,~P_{LT}  = \sum\limits_{i = 1}^{p_1 } {P_{Li} },~ \underline \omega  _{COI}  = \omega _{COI}  - \omega ^* ,~\underline P _{\tau gT}  = P_{\tau gT}  - P_{gT}^* ,~\underline P _{LT}  = P_{LT}  - P_{LT}^* $. Assuming $z_1 = \left[ {\underline \omega  _{COI} ~~\underline P _{\tau gT} } \right]^T ,~~w_1 =  - \underline P _{LT}$, \eqref{eq:ClassB_COI} can be expressed as $\dot{z_1} = g_1(z_1,w_1), ~g_1: \R^{2}\times\R \rightarrow \R^{2}$, where $g_1$ is locally Lipschitz in $(z_1, w_1)$, and $g_1(0, 0) = 0$. We present the following lemma to analyze Lyapunov stability of this system.
\begin{lemma}
For class-B GFCs, the equilibrium $z_1 = 0$ is globally asymptotically stable $\forall \; d_{pgT}, d_{pci}$ $>0$,$\forall i$.
\end{lemma}
\begin{proof}
It is a simple extension of Lemma~\ref{Lemma:ClassBLyap}. We use the Lyapunov function $V_6 = H_{T}\underline{\omega}_{COI}^{2} + \frac{\tau_{gT}}{2d_{pgT}}\underline{P}_{\tau gT}^{2}$ with $d_{pgT}>0$ and notice that $\sum\limits_{i = 1}^{n_1 }{\underline{\omega}_{COI} sat(d_{pci}\underline{\omega}_{COI}, \underline P _{ci}^{\max })} > 0 ~~ \forall \; \underline{\omega}_{COI} \in \R - \left\{ {0} \right\}, d_{pci} > 0~\forall i $, which proves the Lemma.
\end{proof}
Next, we present a corollary relating the input-to-state stability of this system.
\begin{corollary}
Theorem \ref{Th:InputStateClassB} can be extended for establishing the input-to-state stability of \eqref{eq:ClassB_COI} with the following modifications: (1) $\underline P _{c}^{\max }$ and $d_{pc}$ correspond to the minimum value of  ${\underline P _{ci}^{\max } \tanh \left( {\left| {\frac{{d_{pci} }}{{\underline P _{ci}^{\max } }}\underline \omega  _{COI} } \right|} \right)}, \forall i = 1:n_1$, (2) $\chi _1 \left( {\left| w_1 \right|} \right) = \frac{{\underline P _c^{\max } }}{{d_{pc} }}\tanh ^{ - 1} \left( {\frac{{|w_1|}}{{\theta n_1 \underline P _c^{\max } }}} \right) $ and $\chi _2 \left( {\left| w_1 \right|} \right) = \left[ {\frac{|w_1|d_{pgT}}{{\theta }} \chi _1 \left( {\left| w_1 \right|} \right)} \right]^{\frac{1}{2}}$, $~\forall w_1\in (-\theta n_1 \underline P _c^{\max } ,~\theta n_1 \underline P _c^{\max } )$.
\end{corollary} 
\begin{proof}
Assuming $V_6$ as the Lyapunov function, this can be easily proved following same steps as in Theorem~\ref{Th:InputStateClassB}.
\end{proof}

\subsection{Remarks on Assumptions}\label{sec:Remarks}
1. \textit{Network model:} It was shown in \cite{hug_2019_low_inertia} that network dynamics introduces both positive and negative effects on stability in systems with GFCs and SGs. Although algebraic representation of the network gives a conservative stability estimate, the inclusion of network dynamics imposes strict upper bounds on droop feedback gains for ensuring voltage and frequency stability. 

2. \textit{AC current limits:} The ac current limits are used to constrain GFC current during faults. We point out that following a fault, typically the reactive component of current increases significantly \cite{blaabjerg_2019_fault_reactive_current} compared to the real component. As a result, this might not lead to dc-side current saturation. On the other hand, the problem of generation loss leads to increase in real power output and hence dc-side current limit is reached first. Assuming that the available headroom $(i_{dc}^{max} - i_{dc})$ is not very large (which is typical), the ac current limit might not be hit in this condition. If, however this is not the case, then ac-side constraints need to be taken into account in stability analysis, which is outside the scope of the present paper.

3. \textit{Frequency of class-B GFCs:} For class-B GFCs, the working assumption in the $2$-bus test system is $\omega _c  \approx \omega _g$, whereas in the multimachine system, we assume $\omega_{COI} \approx \omega_{ci}, ~\forall i=1,2,\dots,n_1$. In reality, this may not be true. Also, such models cannot capture the oscillatory electromechanical dynamics present in practical multimachine systems that reflects the angle stability issues. 

4. \textit{DC voltage filtering in class-B GFC:} In reality, switching ripple in dc-link voltage can propagate to angle reference through $\omega_c$ of class-B GFCs. If a low pass filter is used to mitigate this issue, it needs to be considered in the stability analysis.
 
\vspace{5pt}
\section{Results \& Discussions}
\vspace{5pt}
For validating the proposed lemmas and theorems, we consider the test system shown in Fig. \ref{fig_test_systeml}. To that end, the averaged models shown in Fig. \ref{fig:block_diagram} are built in Matlab Simulink and a detailed switched model of a standalone GFC connected to a constant power load $P_{LC}$ is developed in EMTDC/PSCAD including the control loops shown in Fig.\ref{fig_control_block}.
\begin{figure}[ht]
\vspace{-5pt}
\includegraphics[width=\linewidth]{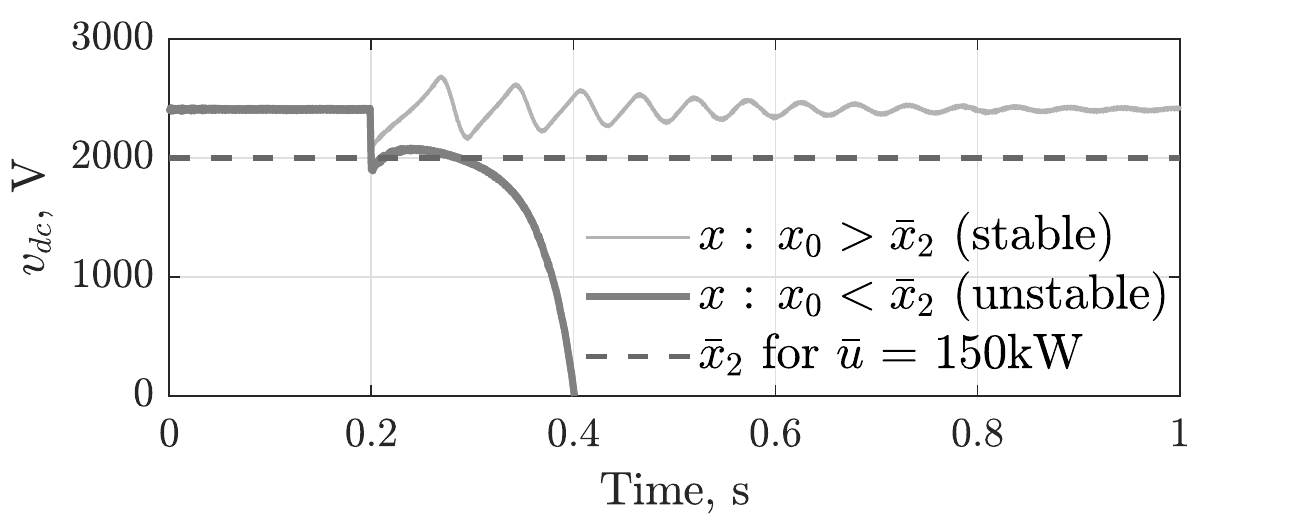}
\caption{Class-A GFC: unforced response from PSCAD model.}
\label{fig_class_a_roa}
\vspace{-5pt}
\end{figure}
\par We validate the ROAs for both classes of GFCs using the PSCAD model by switching the dc bus capacitor voltage to a value $x_0$ at $t = 0.2$ s while operating at equilibrium $(\bar{x}_1, \bar u)$. In Fig. \ref{fig_class_a_roa}, it is shown that for class-A GFC, $v_{dc}$ collapses if $x_0<\bar{x}_2$, whereas it is stable if $x_0>\bar{x}_2$ by slight margin, which validates the ROA defined in Theorem \ref{Th:ClassALyap}. Figure \ref{fig_class_b_roa} shows that $v_{dc}$ returns back to $\bar{x}_1$ even if it is switched below $\bar{x}_2$ for class-B GFC.

\begin{figure}[ht]
\vspace{-0pt}
\includegraphics[width=\linewidth]{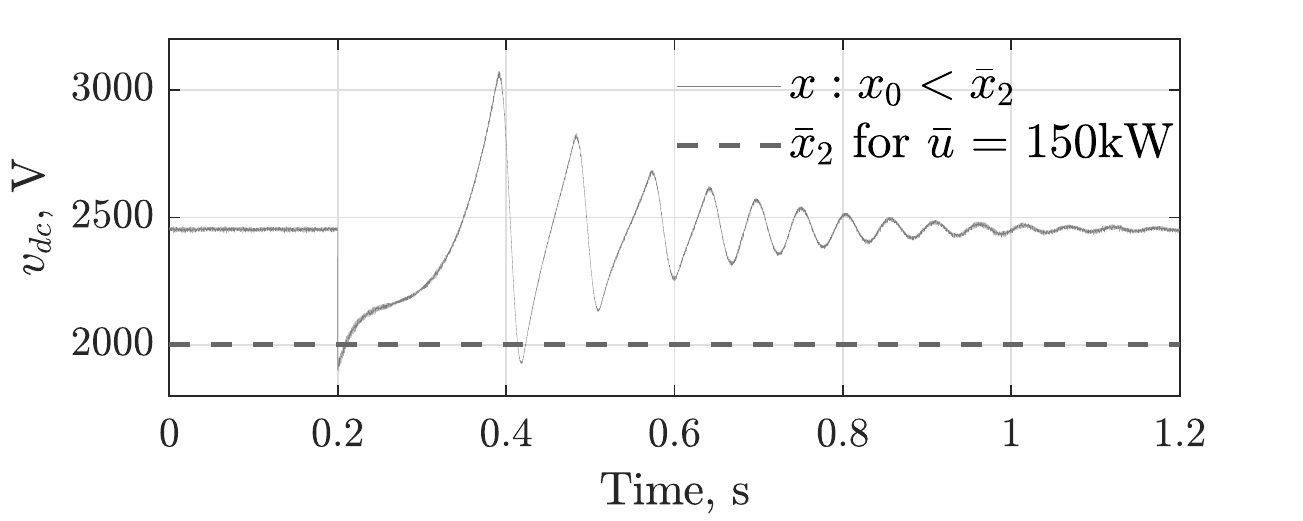}
\vspace{-15pt}
\caption{Class-B GFC: unforced response from PSCAD model.}
\vspace{-7pt}
\label{fig_class_b_roa}
\end{figure}

Next, we validate the ROAs using the average models that capture the dynamics of both SG and GFC. Figure \ref{fig_class_a_roa_avg_model} shows the unforced response of these models by initializing $v_{dc}$ at different values while operating at equilibrium $(\bar{x}_1, \bar u)$. Here, the class-B GFC is stable even when the initial voltage state is significantly lower than $\bar{x}_{2}$.

\begin{figure}[ht]
\vspace{-5pt}
\includegraphics[width=\linewidth]{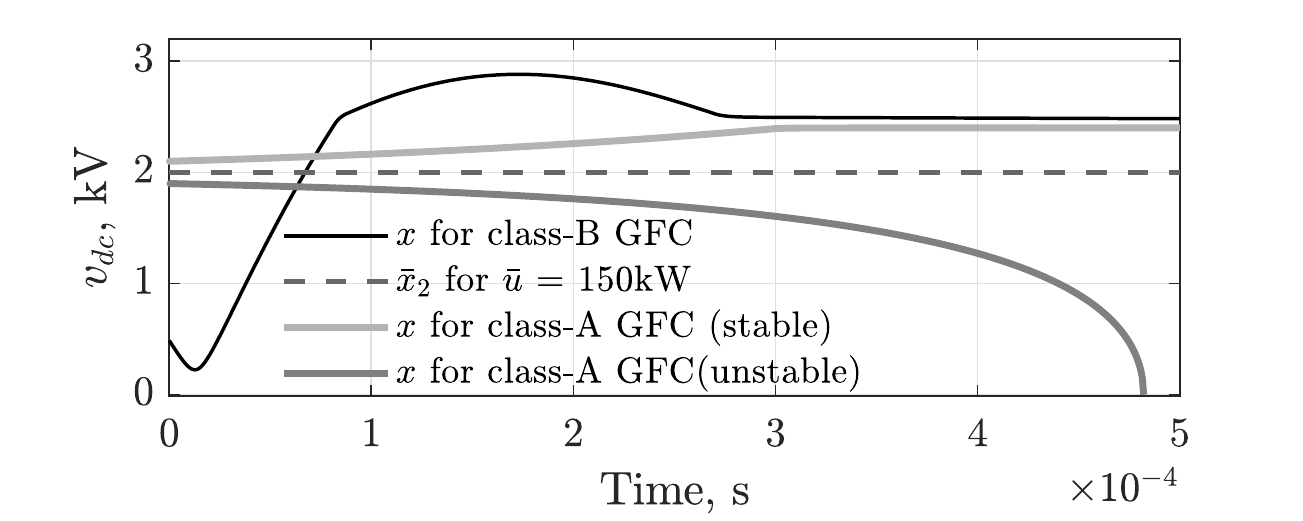}
\caption{Unforced response from averaged models.}
\label{fig_class_a_roa_avg_model}
\vspace{-5pt}
\end{figure}
\par

To validate Theorem \ref{Th:ClassAInOut}, a small step change is given in the load from $\bar{u} = 175$ kW to $\bar{u}_{m} = 177$ kW in PSCAD model of class-A GFC (see, Fig. \ref{fig_class_a}(a)). It can be seen from Fig. \ref{fig_class_a}(b), that the dc voltage is stable. In Fig. \ref{fig_class_a} (c,d), it is shown that when $u = \bar{u}_m$, the unforced response becomes unstable when the initial value of $x$ is less than $x_m$, which proves Theorem \ref{Th:ClassAInstab}.
\begin{figure}[ht]
\vspace{-5pt}

\includegraphics[width=\linewidth]{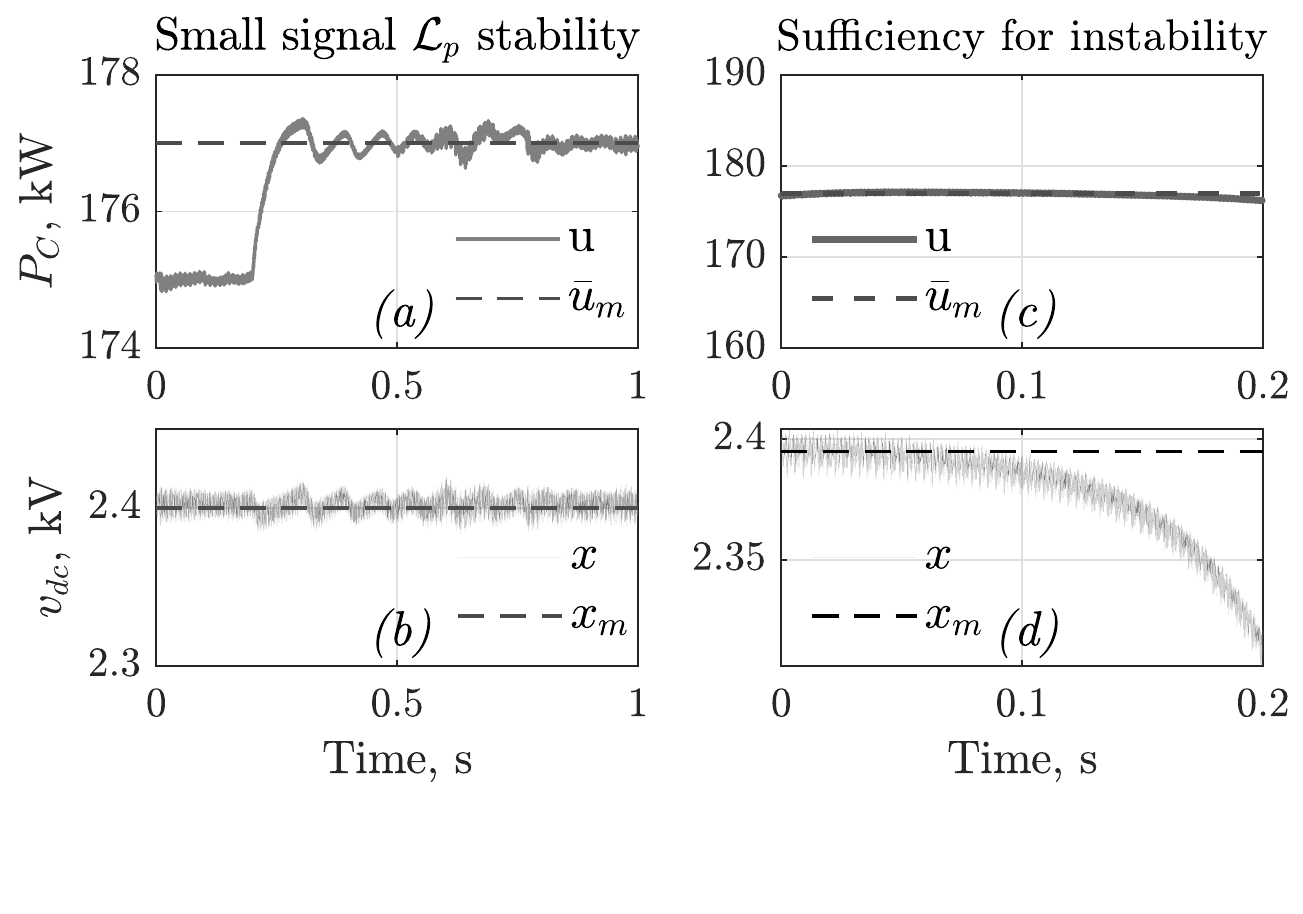}
\vspace{-30pt}
\caption{Class-A GFC PSCAD model: (a),(b): Forced response; (c),(d): Unforced response.}
\vspace{-0pt}

\label{fig_class_a}
\end{figure}
\par
For class-B GFC, a large step change is given in the load from $\bar u = 165$ kW to $\bar u_m = 177$ kW -- Fig. \ref{fig_class_b} confirms the input to state stability per Theorem \ref{Th:InputStateClassB}.

\begin{figure}[ht]
\vspace{-5pt}
\includegraphics[width=\linewidth]{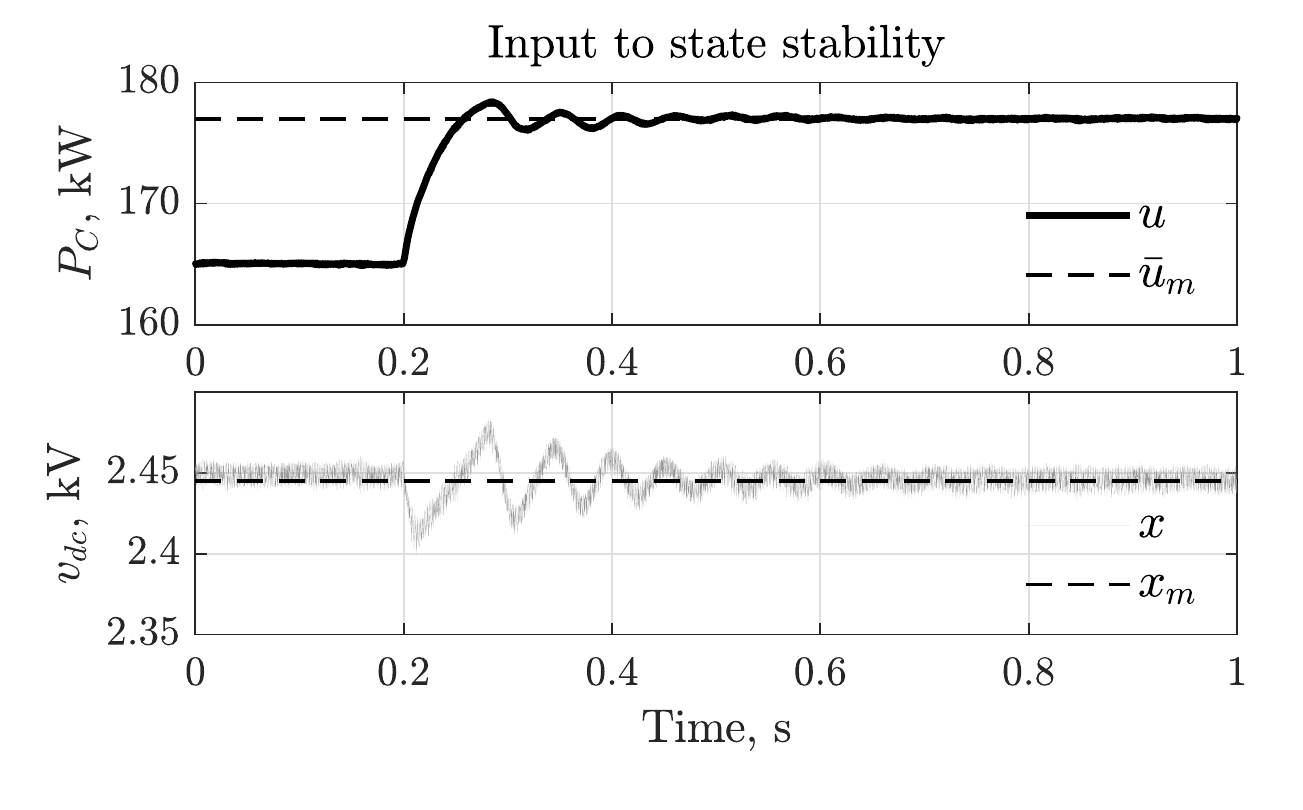}
\vspace{-20pt}
\caption{Power and DC voltage plot for class-B GFC.}
\vspace{-10pt}
\label{fig_class_b}
\end{figure}

\vspace{0pt}
\section{Conclusion}

Stability analysis of power systems consisting of SGs and GFCs with dc-side current limitation showed that the ROA of class-A GFC is a subset of its class-B counterpart. We established the conditions for small-signal finite gain $\mathcal{L}_{p}$ stability of class-A GFC and input-to-state state stability of class-B GFC following a bounded variation in the load of the system, and validated the results through simulation studies.
\section*{Appendix}
\vspace{ 0 pt}
\setlength{\tabcolsep}{6pt} 
\renewcommand{\arraystretch}{1.25}
\begin{center}
\begin{tabular}[ht]{||c c c c c c||}
\hline
\multicolumn{1}{|c|}{$k_c$}     & \multicolumn{1}{c||}{1.6e3 $\mho$}  & \multicolumn{1}{c|}{$v_{dc}^*$}   & \multicolumn{1}{c||}{2.44 kV} & \multicolumn{1}{c|}{$C_c$}   & \multicolumn{1}{c|}{8 mF}   \\ \hline
\multicolumn{1}{|c|}{$G_c$}     & \multicolumn{1}{c||}{0.83 $\mho$}   & \multicolumn{1}{c|}{$P_c^*$}    & \multicolumn{1}{c||}{150 kW}  & \multicolumn{1}{c|}{$P_g^*$}  & \multicolumn{1}{c|}{150 kW} \\ \hline
\multicolumn{1}{|c|}{b}      & \multicolumn{1}{c||}{5e3 p.u.}  & \multicolumn{1}{c|}{$H_g$}     & \multicolumn{1}{c||}{3.7 s}    & \multicolumn{1}{c|}{$\tau_g$} & \multicolumn{1}{c|}{5 s}     \\ \hline
\multicolumn{1}{|c|}{$k_{m}$} & \multicolumn{1}{c||}{128.75 $V^{-1}$} & \multicolumn{1}{c|}{$\omega^*$} & \multicolumn{1}{c||}{314.16 rad/s} & \multicolumn{1}{c|}{$d_{pg}$}  & \multicolumn{1}{c|}{7 p.u.}     \\ \hline
\multicolumn{1}{|c|}{$d_{pc}$}    & \multicolumn{1}{c||}{1e-3 p.u.}   & \multicolumn{1}{c|}{$i_{dc}^{max}$}       & \multicolumn{1}{c||}{75 A}       & \multicolumn{1}{c|}{$P_{c}^{max}$}     & \multicolumn{1}{c|}{178 kW}      \\ \hline
\multicolumn{1}{l}{}         & \multicolumn{1}{l}{}        & \multicolumn{1}{l}{}        & \multicolumn{1}{l}{}        & \multicolumn{1}{l}{}      & \multicolumn{1}{l}{}       \\ 
\end{tabular}
\end{center}
\vspace{-0 pt}
\renewcommand*{\bibfont}{\footnotesize}

\printbibliography

\end{document}